\pdfoutput=1
\documentclass[11pt]{article}

\usepackage[margin=1in]{geometry}
\usepackage{amsmath,amssymb,amsthm}
\usepackage{mathtools}
\usepackage[authoryear,round]{natbib}
\usepackage{setspace}
\usepackage{enumitem}
\usepackage[hidelinks]{hyperref}
\usepackage{microtype}
\usepackage{pgfplots}
\pgfplotsset{compat=1.16}
\usetikzlibrary{arrows.meta,patterns}
\pgfplotsset{calibfig/.style={
  width=0.46\textwidth, height=4.4cm, axis on top,
  enlargelimits=false, tick align=outside,
  scaled y ticks=false, scaled x ticks=false,
  tick label style={font=\footnotesize},
  label style={font=\footnotesize},
  every axis plot/.append style={line width=0.8pt},
}}

\newtheorem{proposition}{Proposition}
\newtheorem{lemma}{Lemma}
\newtheorem{corollary}{Corollary}
\theoremstyle{definition}

\theoremstyle{remark}
\newtheorem{remark}{Remark}

\newcommand{\ind}{\mathbf{1}}
\newcommand{\Dl}{D_{L}}
\newcommand{\Dh}{D_{H}}

\DeclareMathOperator{\E}{\mathbb{E}}
\newcommand{\SEP}{(\textnormal{SEP})}

\title{\textbf{The Accommodation Trap:\\ Survival Dependence, Communication, and the\\ Allocation of Invisible Work}}
\author{Samiha Tariq\\[4pt]
\normalsize School of Analytics, Finance, and Economics\\
\normalsize Southern Illinois University Carbondale\\
\normalsize Carbondale, Illinois, USA}
\date{\today}

\begin{document}
\maketitle

\begin{abstract}
\noindent This paper develops a theory of \emph{accommodation traps} in workplace hierarchies. Workers with weak outside options face a higher relational cost of appearing unavailable, resistant, or difficult, and so adopt accommodative communication (extra deference, softened boundaries, visible flexibility) to preserve the employment relationship. This short-run strategy is also an informative signal: a friction-minimizing supervisor rationally infers that an accommodating worker is less likely to resist, and assigns that worker the invisible, low-promotability work. The signal harms the worker who sends it. Unlike agents in ratchet-effect models, who conceal information that would invite heavier demands, dependent workers reveal their low refusal cost, because the relational value of accommodating rises with dependence. The resulting inequality requires no productivity difference and no taste for discrimination, and the supervisor's beliefs are correct, so correcting beliefs cannot remove it. When visible work builds outside options and invisible work erodes them, the allocation becomes self-reinforcing, and even workers who start out identical are locked into unequal roles with probability one. The equilibrium is privately rational for every party yet socially inefficient whenever diminishing returns to outside options outweigh the supervisor's friction cost. Attaching recognition to invisible work, rotating it by rule, or limiting managerial discretion each weakens or breaks the link from communication to assignment that sustains the trap.


\end{abstract}

\section{Introduction}

Workers do not enter the workplace with equal fallback options. Some can decline a request, set a boundary, or keep a reply short because they have savings, a partner's income, secure immigration status, or a strong outside offer. Others cannot easily risk displeasing a supervisor: their job is tied to funding, a visa, health insurance, a degree, or basic solvency. This paper asks whether that asymmetry in \emph{dependence} shapes how workers communicate, and whether communication in turn shapes which worker is handed the organization's invisible work.

The motivating observation is a familiar workplace asymmetry. A worker who depends heavily on a job tends to communicate differently from one with strong outside options: more accommodating, more apologetic, quicker to agree, slower to refuse. This is individually rational: it protects the relationship the worker cannot afford to lose. But the same behavior is observable, and a supervisor who must get unglamorous work done can learn from it. An accommodating worker looks like a low-friction worker: easier to ask, less likely to push back, cheaper to assign. The worker's strategy for short-run security can therefore expose them to long-run disadvantage. This is the \emph{accommodation trap}.

The phenomenon of unequal invisible-work allocation is documented. \citet{babcock2017} show that women volunteer for, are asked to do, and accept low-promotability tasks more than men, with beliefs about who will say ``yes'' a central driver. \citet{baewoolley} show that managers pile extra work onto intrinsically motivated employees, mediated by a na\"ive belief that such workers will enjoy the extra work. The contribution of this paper is not to re-establish that some workers absorb more invisible work. It is to identify a distinct \emph{mechanism} that explains who becomes the predictable recipient of that work and why the pattern persists; the specific contributions are listed at the end of this section.

\paragraph{What the model says.} Two workers differ in privately known dependence. Each chooses a communication style, bounded or accommodative, before a supervisor allocates one visible (career-building) task and one invisible (necessary, low-recognition) task. Two primitives do the work, and both flow from a single idea: \emph{the cost of friction with the supervisor rises with dependence}. First, seeming ``difficult'' is more dangerous the more a worker depends on the relationship, so accommodation yields a relational benefit increasing in dependence. Second, a dependent worker assigned invisible work is less likely to resist, because refusal is costly to them; so they are the low-friction target for that work. The supervisor, minimizing expected friction, wants to place invisible work on the more dependent worker but cannot observe dependence directly.

Communication closes that gap. In the separating equilibrium characterized below, dependent workers accommodate and independent workers do not; the supervisor correctly reads accommodation as a marker of dependence and assigns invisible work accordingly. The worker's propensity to accommodate is increasing in dependence (Proposition~\ref{prop:dep}), and invisible work flows to the accommodator (Proposition~\ref{prop:alloc}). A single threshold in dependence separates the two communication styles, and Proposition~\ref{prop:regimes} maps the parameter regions in which the economy sorts (the trap), pools on boundedness (no accommodation), or pools on accommodation.

Embedding the stage game in time, with visible work raising future outside options and invisible work eroding them, the allocation becomes self-reinforcing: the dependence gap widens each period and the configuration is absorbing (Proposition~\ref{prop:trap}). This is the trap proper: a transitory difference becomes permanent. Although every party behaves rationally and no output is lost statically, the allocation is inefficient whenever diminishing returns to outside options outweigh the supervisor's friction cost, because the supervisor internalizes friction but not the worker's welfare from human-capital accumulation (Proposition~\ref{prop:welfare}). Finally, three interventions break the trap by severing the signal-to-allocation link: attaching recognition or promotability to invisible work (Proposition~\ref{prop:recognition}), and rotating invisible work by rule or otherwise limiting discretion (Proposition~\ref{prop:rotation}). Figure~\ref{fig:loop} traces the full loop and marks where each result enters.

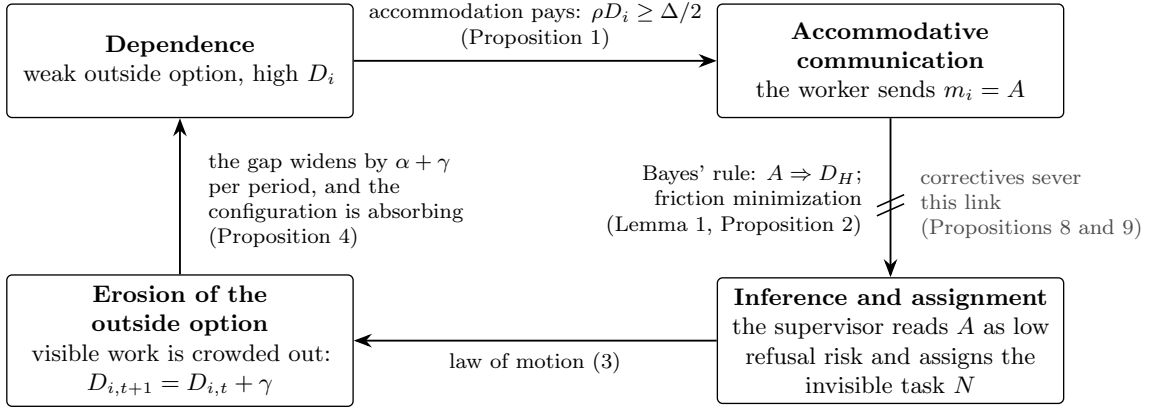
\begin{figure}[t]
\centering
\begin{tikzpicture}[
  font=\footnotesize,
  stage/.style={draw, rounded corners=2pt, line width=0.5pt, align=flush center,
    text width=4.3cm, minimum height=1.5cm, inner sep=4pt},
  flow/.style={-{Stealth[length=2.4mm]}, line width=0.7pt},
  lab/.style={font=\scriptsize, align=center, inner sep=2pt}]
\node[stage] (dep) at (0,0)      {\textbf{Dependence}\\ weak outside option, high $D_i$};
\node[stage] (acc) at (9.4,0)    {\textbf{Accommodative communication}\\ the worker sends $m_i=A$};
\node[stage] (asg) at (9.4,-3.7) {\textbf{Inference and assignment}\\ the supervisor reads $A$ as low refusal risk and assigns the invisible task $N$};
\node[stage] (ero) at (0,-3.7)   {{\bfseries Erosion of the\\ outside option}\\ visible work is crowded out:\\ $D_{i,t+1}=D_{i,t}+\gamma$};
\draw[flow] (dep) -- (acc);
\draw[flow] (acc) -- (asg);
\draw[flow] (asg) -- (ero);
\draw[flow] (ero) -- (dep);
\node[lab, anchor=south] at (4.7,0.08) {accommodation pays: $\rho D_i\ge\Delta/2$\\ (Proposition~\ref{prop:dep})};
\node[lab, anchor=north] at (4.7,-3.78) {law of motion \eqref{eq:lom}};
\node[lab, anchor=east, align=right] at (9.1,-1.85) {Bayes' rule: $A\Rightarrow\Dh$;\\ friction minimization\\ (Lemma~\ref{lem:super}, Proposition~\ref{prop:alloc})};
\node[lab, anchor=west, align=left] at (0.3,-1.85) {the gap widens by $\alpha+\gamma$\\ per period, and the\\ configuration is absorbing\\ (Proposition~\ref{prop:trap})};
\draw[line width=0.6pt] (9.22,-1.93) -- (9.58,-1.77) (9.22,-2.05) -- (9.58,-1.89);
\node[lab, anchor=west, align=left, text=black!70] at (9.72,-1.9) {correctives sever\\ this link\\ (Propositions~\ref{prop:recognition} and~\ref{prop:rotation})};
\end{tikzpicture}
\caption{The accommodation trap. Dependence makes accommodation privately optimal (Proposition~\ref{prop:dep}); the supervisor correctly reads accommodation as a marker of dependence and assigns the invisible task to the accommodator (Lemma~\ref{lem:super}, Proposition~\ref{prop:alloc}); invisible work crowds out visible, career-building work and raises dependence \eqref{eq:lom}, which closes the loop (Proposition~\ref{prop:trap}). Recognition of invisible work (Proposition~\ref{prop:recognition}) and rule-based rotation or reduced discretion (Proposition~\ref{prop:rotation}) break the loop by severing the link from communication to assignment; recognition also weakens the erosion step by reducing $\gamma$.}
\label{fig:loop}
\end{figure}

\paragraph{Contributions.} Relative to these literatures, the paper makes four contributions.
\begin{enumerate}[label=(\arabic*),itemsep=2pt]
\item \emph{A new channel.} How a worker communicates tells the supervisor how costly refusing would be for that worker. The operative trait is dependence on the job, not gender \citep{babcock2017}, intrinsic motivation \citep{baewoolley}, or a supervisor's bias, and the informativeness of communication is an equilibrium object rather than an assumption (Propositions~\ref{prop:dep} and~\ref{prop:alloc}).
\item \emph{A signal that harms its sender.} In ratchet-effect models \citep{fgt1985,laffonttirole1988}, an agent whose low cost would invite heavier future demands has an incentive to conceal it. Here dependent workers reveal their low refusal cost even though doing so attracts the invisible work, because the relational value of accommodating, $\rho D$, rises with dependence while the added risk of invisible work, $\Delta/2$, does not. The same trait that makes accommodation privately valuable makes the accommodator the supervisor's cheapest target (Lemma~\ref{lem:super}). This self-harming separation does not appear to have been formalized for the allocation of invisible work.
\item \emph{Inequality without bias.} The supervisor has no taste for discrimination and her beliefs are correct in equilibrium, so interventions that work by correcting beliefs have nothing to correct. This separates the mechanism from accounts in which a managerial misperception \citep{baewoolley} or a subtle bias \citep{pikulinaferreira2026} drives the gap, and it directs policy toward incentives and discretion: recognition of invisible work, rule-based rotation, and limits on discretion (Propositions~\ref{prop:recognition} and~\ref{prop:rotation}).
\item \emph{Lock-in from equal starts.} Because visible work builds outside options and invisible work erodes them, the allocation entrenches itself (Proposition~\ref{prop:trap}), and even workers who start out identical are locked into unequal roles with probability one (Proposition~\ref{prop:lockin}). The logic echoes the self-confirming stereotypes of \citet{coateloury1993}, but it operates worker by worker, without group identity or an investment decision.
\end{enumerate}

\paragraph{Roadmap.} Section~\ref{sec:lit} positions the paper. Section~\ref{sec:model} sets up the stage game. Section~\ref{sec:equilibrium} characterizes equilibrium and comparative statics. Section~\ref{sec:dynamics} develops the dynamic trap. Section~\ref{sec:welfare} treats welfare and correctives, and Section~\ref{sec:calibration} gives a numerical illustration. Section~\ref{sec:discussion} discusses scope and interpretation, Section~\ref{sec:empirics} sketches empirical designs, and Section~\ref{sec:conclusion} concludes.

\section{Related Literature}\label{sec:lit}

The accommodation trap sits where three lines of research meet: work on who performs unrewarded tasks, work on how supervisors learn about the people they manage, and work on why differences between workers harden into lasting inequality. Each explains part of the pattern studied here. Read together, they point to a missing link: a single source of heterogeneity that shapes what a worker says, what the supervisor infers, and how the resulting allocation feeds back on the worker's position. The discussion below is organized around these three questions.

\paragraph{Who does the unrewarded work.} Since \citet{daniels1987} drew attention to invisible work, the question has shifted from whether such work exists to who is expected to do it. Two recent answers share a common structure. \citet{babcock2017} show that women are asked more often, volunteer more often, and accept more low-promotability tasks than men, driven in part by shared beliefs that women will say yes. \citet{baewoolley} show that managers give extra tasks to employees they see as intrinsically motivated, because they na\"ively believe these employees will enjoy the extra work. In both accounts the allocation follows a belief about how readily a worker will accept a burden, and the belief attaches to a visible trait: gender in one case, apparent motivation in the other. Neither account explains where such a belief comes from when no trait is visible, and in the second it is explicitly a misperception. Labor economics suggests an economic source for the willingness to accept: workers' outside options differ, and they matter for pay \citep{caldwelldanieli}. This paper connects the two literatures. Outside options shape not only what a worker can bargain for but also how the worker communicates, and communication shapes which worker the supervisor asks; the belief that drives the allocation is generated by the worker's own behavior and is correct in equilibrium. The question also differs from the organizational-design literature on task assignment, which asks how tasks should be grouped into jobs or units \citep{puschke2009} rather than which individual receives a given unrewarded task.

\paragraph{What supervisors learn from workers.} If allocation depends on beliefs, the next question is how those beliefs form. The classic answers have the supervisor learn from an observable proxy \citep{phelps1972} or from information the firm gathers itself, for example by testing workers before assigning them \citep{maluegxu1997}. When workers can shape the information, the answer depends on what revealing costs them. In \citet{spence1973}, revealing a high type pays, so high types signal. In the ratchet-effect literature, revealing a low cost invites more demanding terms later, so efficient agents conceal it and equilibria involve substantial pooling \citep{fgt1985,laffonttirole1988}. Communication inside hierarchies adds its own distortions: subordinates who are evaluated subjectively tell supervisors what the supervisors already believe \citep{prendergast1993}, manage impressions to win favorable ratings \citep{wayneliden1995}, or stay silent about problems \citep{morrisonmilliken2000}. Accommodation occupies an unusual position in this landscape. Like the ratchet agent, the accommodating worker reveals something that invites a heavier load. Unlike the ratchet agent, the worker reveals it anyway, because the relational return to accommodating, the value of keeping an informal relationship in good standing \citep{bgm2002}, rises with the very dependence it reveals. And unlike the ingratiating or silent subordinate, whose communication distorts information, the accommodating worker's message is informative in equilibrium. The outcome is a separating equilibrium in which the signal harms its sender. The term \emph{accommodation signaling} refers here to communication of this kind, which lowers a supervisor's perceived risk of refusal and of friction.

\paragraph{How gaps persist.} A third literature asks why differences between workers arise and last. One family of answers traces unequal treatment to preferences or bias: supervisors favor some subordinates over others \citep{prendergasttopel}, relational feelings between supervisors and workers distort evaluations \citep{dugheramarciano2022}, and subtle, hard-to-verify biases in promotion decisions generate large gaps in skills and promotion outcomes \citep{pikulinaferreira2026}. A second family shows that inequality can persist without bias, because beliefs become self-confirming: in \citet{coateloury1993}, employers' beliefs about groups shape workers' incentives to invest, so groups that are identical ex ante can end up unequal ex post, with employers' beliefs about them correct. The dynamic result of this paper belongs to the second family but works through a different channel. There is no bias and no investment decision; persistence runs through the allocation of work itself, which erodes the outside options of the worker who receives the invisible task and so reproduces the dependence that led to it. Because the feedback operates worker by worker, it needs no group identity, and it locks in even workers who start out identical. Patience matters here as it does in implicit contracting, where a lower discount factor can move play from pooling to separation \citep{gurtler2008}: the short horizon of a dependent worker helps keep the trap in place (Section~\ref{sub:fl}).

\paragraph{Where the paper fits.} Taken together, these literatures explain why unrewarded work may follow visible traits, how supervisors may learn about workers, and why gaps may persist through bias or self-confirming beliefs, but they largely treat these as separate problems. The accommodation trap joins them through one primitive: dependence raises the cost of friction with the supervisor. That single force makes accommodation privately valuable, makes the accommodating worker the cheapest person to burden, and, through the allocation of work, reproduces itself. The contributions listed in the Introduction follow from this link.

\section{The Model}\label{sec:model}

\paragraph{Players and types.} A supervisor and two workers, $i\in\{1,2\}$. Worker $i$ privately observes a \emph{dependence} type $D_i\in\{\Dl,\Dh\}$ with $0<\Dl<\Dh$. Types are independent across workers with $\Pr(D_i=\Dh)=\pi\in(0,1)$. Dependence is the inverse of the outside option: a higher $D_i$ means a weaker fallback and more to lose if the relationship sours. The supervisor knows only the distribution.

\paragraph{Tasks.} Two tasks must be done and are allocated one to each worker: a \emph{visible} task $V$ (career-building, recognized, promotion-relevant) and an \emph{invisible} task $N$ (necessary, routine, low-recognition). Thus an allocation is a bijection: one worker does $V$, the other does $N$.

\paragraph{Worker task payoffs.} A worker obtains $w_V$ from $V$ and $w_N$ from $N$, with
\[
\Delta \;\equiv\; w_V-w_N \;>\;0 ,
\]
the worker's value of receiving the visible rather than the invisible task.

\paragraph{Communication.} Before allocation, workers simultaneously choose a communication style $m_i\in\{A,B\}$: \emph{accommodative} ($A$) or \emph{bounded} ($B$). Accommodation bundles the behaviors of the dependent worker: deference, gratitude, softened boundaries, visible flexibility. Boundedness is concise and boundary-setting.

\paragraph{Relational payoff of communication.} Setting boundaries risks the supervisor perceiving a worker as difficult, which carries an expected relational loss (a worse reference, lower renewal probability, withdrawn goodwill) that bites in proportion to how much the worker depends on the job. Accommodation neutralizes this. In reduced form, relative to $B$, choosing $A$ yields worker $i$ a relational benefit
\[
g(D_i)\;=\;\rho\,D_i,\qquad \rho>0,
\]
increasing in dependence and independent of the task received. (Remark~\ref{rem:microfound} discusses micro-foundation.)

\paragraph{Supervisor payoff: friction.} Both tasks are always completed, so output does not vary with the allocation. What varies is friction. A worker assigned the invisible task may resist (push back, comply grudgingly, or refuse), imposing friction cost $\phi>0$ on the supervisor. The probability of resistance is $r(D_i)$, \emph{strictly decreasing} in dependence:
\[
r(\Dl)=r_L \;>\; r(\Dh)=r_H\;\ge 0 .
\]
A more dependent worker resists less, because refusal endangers the relationship they cannot afford to lose. Note that $g$ increasing and $r$ decreasing flow from the same primitive: \emph{dependence raises the cost of friction with the supervisor}. The supervisor minimizes expected friction; they therefore prefer to place $N$ on the worker they believe more dependent.

\paragraph{Timing.}
\begin{enumerate}[label=(\arabic*),itemsep=1pt]
  \item Nature draws $(D_1,D_2)$, privately observed.
  \item Workers simultaneously choose $(m_1,m_2)$.
  \item The supervisor observes $(m_1,m_2)$, forms beliefs, and allocates $N$ to one worker and $V$ to the other to minimize expected friction.
  \item Payoffs realize.
\end{enumerate}
\noindent Figure~\ref{fig:timing} depicts this sequence, together with the state update used in the dynamic extension of Section~\ref{sec:dynamics}.

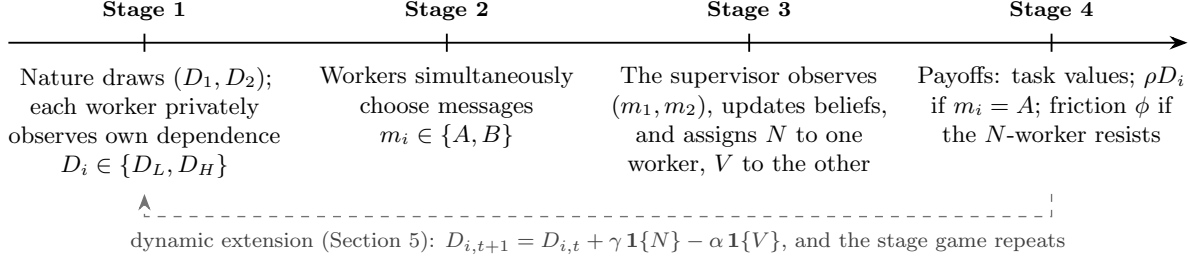
\begin{figure}[t]
\centering
\begin{tikzpicture}[font=\footnotesize,
  desc/.style={align=flush center, text width=3.6cm, anchor=north, inner sep=2pt},
  stg/.style={font=\scriptsize\bfseries, anchor=south}]
\draw[-{Stealth[length=2.4mm]}, line width=0.7pt] (0,0) -- (15.6,0);
\foreach \x/\s in {1.8/1, 5.8/2, 9.8/3, 13.8/4} {
  \draw[line width=0.7pt] (\x,-0.12) -- (\x,0.12);
  \node[stg] at (\x,0.16) {Stage \s};
}
\node[desc] at (1.8,-0.24) {Nature draws $(D_1,D_2)$; each worker privately observes own dependence $D_i\in\{\Dl,\Dh\}$};
\node[desc] at (5.8,-0.24) {Workers simultaneously choose messages $m_i\in\{A,B\}$};
\node[desc] at (9.8,-0.24) {The supervisor observes $(m_1,m_2)$, updates beliefs, and assigns $N$ to one worker, $V$ to the other};
\node[desc] at (13.8,-0.24) {Payoffs: task values; $\rho D_i$ if $m_i=A$; friction $\phi$ if the $N$-worker resists};
\draw[-{Stealth[length=2.2mm]}, dashed, line width=0.5pt, black!55] (13.8,-2.0) -- (13.8,-2.3) -- (1.8,-2.3) -- (1.8,-2.0);
\node[font=\scriptsize, text=black!75, anchor=north] at (7.8,-2.36) {dynamic extension (Section~\ref{sec:dynamics}): $D_{i,t+1}=D_{i,t}+\gamma\,\ind\{N\}-\alpha\,\ind\{V\}$, and the stage game repeats};
\end{tikzpicture}
\caption{Timing of the stage game. In the dynamic model, each worker's dependence is updated by the task received, according to \eqref{eq:lom}, before the next period's stage game.}
\label{fig:timing}
\end{figure}

\paragraph{Strategies and equilibrium.} A worker strategy is $\sigma_i:\{\Dl,\Dh\}\to\Delta\{A,B\}$. A supervisor strategy is an allocation rule $a:\{A,B\}^2\to[0,1]$, where $a(m_1,m_2)$ is the probability that worker~$1$ receives $N$. Beliefs $\mu$ assign, to each observed message, a posterior over the sender's type. The solution concept is Perfect Bayesian Equilibrium (PBE): workers and supervisor are sequentially rational, and $\mu$ is derived from strategies by Bayes' rule wherever possible.

\begin{remark}\label{rem:condition}
The supervisor has a strict motive to sort only if $r_L>r_H$. If $r_L=r_H$, friction is independent of dependence, the supervisor is indifferent over which worker does $N$, and communication carries no allocative consequence. The analysis assumes $r_L>r_H$ throughout; this is the natural case and is precisely what makes dependence allocatively relevant.
\end{remark}

\section{Equilibrium}\label{sec:equilibrium}

This section first records the supervisor's behavior, then the signal-induced assignment probabilities, and then characterizes equilibrium.

\begin{lemma}[Friction-minimizing allocation]\label{lem:super}
Fix the supervisor's posterior beliefs. The supervisor assigns $N$ to the worker with the higher posterior probability of being type $\Dh$ (equivalently, the lower posterior $\E[r(D)]$), breaking ties by a fair coin. In particular, under the conjectured separating play $\{\Dh\mapsto A,\ \Dl\mapsto B\}$:
\[
a(A,B)=1,\qquad a(B,A)=0,\qquad a(A,A)=a(B,B)=\tfrac12 .
\]
\end{lemma}

\begin{proof}
Expected friction from assigning $N$ to worker $i$ is $\phi\,\E[r(D_i)\mid m_i]$. Since $r$ is strictly decreasing, $\E[r(D)\mid m]$ is decreasing in the posterior probability of $\Dh$. Assigning $V$ generates no friction. The supervisor therefore places $N$ on the worker with the higher posterior $\Pr(\Dh\mid m)$ and is indifferent under equal posteriors. Under separating beliefs, $A\Rightarrow\Dh$ and $B\Rightarrow\Dl$; with $r_H<r_L$ the stated rule follows.
\end{proof}

\begin{lemma}[Signal-induced assignment probabilities]\label{lem:probs}
Suppose worker $i$'s opponent plays the separating strategy (sends $A$ iff type $\Dh$, so sends $A$ with probability $\pi$), and the supervisor allocates as in Lemma~\ref{lem:super}. Then the probability that worker $i$ is assigned the invisible task is
\[
\Pr(N\mid m_i=A)=1-\tfrac{\pi}{2},
\qquad
\Pr(N\mid m_i=B)=\tfrac{1-\pi}{2},
\]
so that sending $A$ raises the probability of invisible work by
\[
q\;\equiv\;\Pr(N\mid A)-\Pr(N\mid B)\;=\;\tfrac12 .
\]
\end{lemma}

\begin{proof}
Condition on the opponent's message, which is $A$ with probability $\pi$ and $B$ with probability $1-\pi$. If worker $i$ sends $A$: against an $A$-opponent the profile is $(A,A)$, giving $i$ the task $N$ with probability $\tfrac12$; against a $B$-opponent the profile is $(A,B)$, giving $i$ task $N$ with probability $1$. Hence $\Pr(N\mid A)=\pi\cdot\tfrac12+(1-\pi)\cdot1=1-\tfrac{\pi}{2}$. If worker $i$ sends $B$: against an $A$-opponent the profile is $(B,A)$, giving $i$ task $N$ with probability $0$; against a $B$-opponent the profile is $(B,B)$, probability $\tfrac12$. Hence $\Pr(N\mid B)=\pi\cdot0+(1-\pi)\cdot\tfrac12=\tfrac{1-\pi}{2}$. Subtracting gives $q=\tfrac12$.
\end{proof}

\noindent The clean value $q=\tfrac12$ is an artifact of the symmetric two-worker, binary-type structure; what matters for the results is $q>0$ (Remark~\ref{rem:general}).

\subsection{The separating equilibrium}

A worker of type $D$ who sends $A$ obtains the relational benefit $\rho D$ and expected task payoff $w_V-\Delta\,\Pr(N\mid A)$; sending $B$ forgoes the relational benefit and gives $w_V-\Delta\,\Pr(N\mid B)$. The net gain from accommodating is
\begin{equation}\label{eq:netgain}
\Phi(D)\;\equiv\;\underbrace{\rho D}_{\text{relational benefit}}\;-\;\underbrace{\Delta\, q}_{\text{extra invisible-work risk}}
\;=\;\rho D-\tfrac{\Delta}{2}.
\end{equation}
Because $\Phi$ is strictly increasing in $D$, accommodation obeys a threshold rule: worker accommodates iff $D\ge D^{*}$, where
\begin{equation}\label{eq:threshold}
D^{*}\;\equiv\;\frac{\Delta\,q}{\rho}\;=\;\frac{\Delta}{2\rho}.
\end{equation}

\begin{proposition}[Dependence induces accommodation]\label{prop:dep}
Suppose
\begin{equation*}
\rho\,\Dl<\tfrac{\Delta}{2}\le\rho\,\Dh,
\qquad\text{equivalently}\qquad
\Dl<D^{*}\le\Dh. \tag*{\SEP}
\end{equation*}
Then there is a Perfect Bayesian Equilibrium in which type $\Dh$ accommodates and type $\Dl$ is bounded, the supervisor allocates as in Lemma~\ref{lem:super}, and beliefs are $A\Rightarrow\Dh$, $B\Rightarrow\Dl$. The propensity to accommodate $\Phi(D)=\rho D-\tfrac{\Delta}{2}$ is strictly increasing in dependence.
\end{proposition}

\begin{proof}
Take the conjectured strategies and beliefs. Worker optimality: by \eqref{eq:netgain} and Lemma~\ref{lem:probs}, $\Phi(\Dh)=\rho\Dh-\tfrac{\Delta}{2}\ge0$ under~\SEP, so $\Dh$ weakly prefers $A$; $\Phi(\Dl)=\rho\Dl-\tfrac{\Delta}{2}<0$, so $\Dl$ strictly prefers $B$. Supervisor optimality is Lemma~\ref{lem:super}. Because both messages occur with positive probability (each type is realized with positive probability and plays a distinct message), every message profile is on the equilibrium path, so beliefs are pinned down by Bayes' rule: a worker sending $A$ is type $\Dh$ and one sending $B$ is type $\Dl$. All conditions of PBE hold. Monotonicity of $\Phi$ is immediate since $\rho>0$.
\end{proof}

\begin{proposition}[Accommodation attracts invisible work]\label{prop:alloc}
In the equilibrium of Proposition~\ref{prop:dep}, whenever the two workers send different messages the invisible task is assigned to the accommodator, and unconditionally
\[
\Pr(N\mid A)=1-\tfrac{\pi}{2}\;>\;\tfrac{1-\pi}{2}=\Pr(N\mid B).
\]
Equivalently, the more dependent worker bears the invisible task with strictly higher probability.
\end{proposition}

\begin{proof}
Immediate from Lemmas~\ref{lem:super} and \ref{lem:probs}: on profile $(A,B)$ the accommodator receives $N$ with probability one, and $\Pr(N\mid A)-\Pr(N\mid B)=q=\tfrac12>0$. Since $\Dh$ sends $A$ and $\Dl$ sends $B$ in equilibrium, the dependent worker bears $N$ with the higher probability.
\end{proof}

\subsection{Regimes and comparative statics}

The threshold $D^{*}=\Delta/2\rho$ partitions the parameter space into three regimes.

\begin{proposition}[Three regimes]\label{prop:regimes}
Fix off-path beliefs by the criterion that a worker is more likely to send $A$ the more dependent they are (so an unexpected $A$ is read as $\Dh$). Then:
\begin{enumerate}[label=(\roman*),itemsep=1pt]
\item \emph{Separation (the trap regime).} If $\Dl<D^{*}\le\Dh$, the separating equilibrium of Proposition~\ref{prop:dep} obtains: dependent workers accommodate, independent workers do not, and invisible work sorts onto the dependent.
\item \emph{Pooling on boundedness.} If $D^{*}>\Dh$ (e.g.\ strong outside options across the board, large $\Delta$, or small $\rho$), both types send $B$; signals are uninformative and $N$ is assigned by coin flip.
\item \emph{Pooling on accommodation.} If $D^{*}<\Dl$ (relational stakes dominate for all), both types send $A$; signals are uninformative and $N$ is assigned by coin flip.
\end{enumerate}
\end{proposition}

\begin{proof}
Region (i) is Proposition~\ref{prop:dep}. For (ii), suppose both types send $B$. Beliefs on $B$ equal the prior; an off-path deviation to $A$ is read as $\Dh$ by the stated criterion, so against a $B$-opponent the deviator faces profile $(A,B)$ and receives $N$ with probability one, yielding $\rho D+w_V-\Delta$ versus the equilibrium $w_V-\Delta/2$ from $(B,B)$. The deviation is unprofitable iff $\rho D\le\Delta/2$, which holds for both types when $D^{*}>\Dh$. Hence pooling on $B$ is a PBE. Symmetrically for (iii): if both send $A$, an off-path $B$ is read as $\Dl$, so the deviator (against an $A$-opponent, profile $(B,A)$) avoids $N$ entirely, giving $w_V$ versus the equilibrium $\rho D+w_V-\Delta/2$ from the $(A,A)$ tie; the deviation is unprofitable iff $\rho D\ge\Delta/2$, which holds for both types when $D^{*}<\Dl$. Hence pooling on $A$ is a PBE. Appendix~\ref{app:regimes} gives the full argument.
\end{proof}

\noindent Figure~\ref{fig:regimes} maps the three regimes, both in the space of types $(\Dl,\Dh)$ and in the space of the two stakes $(\rho,\Delta)$.

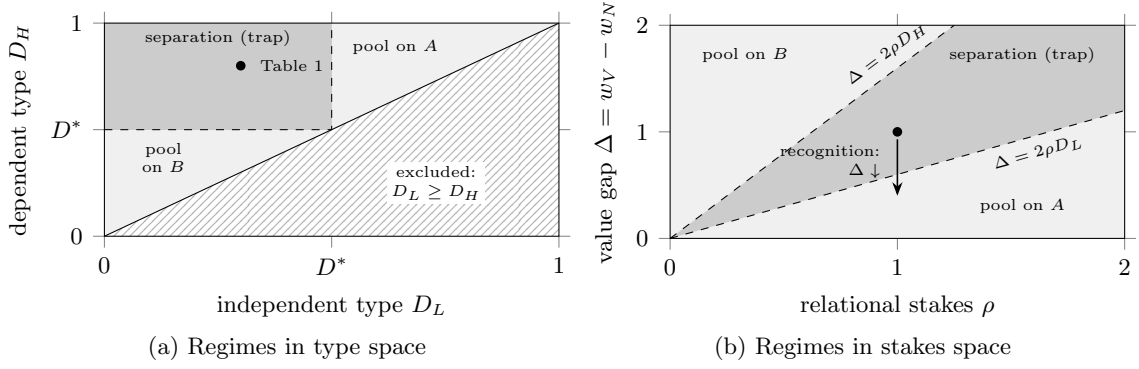
\begin{figure}[t]
\centering
\setlength{\tabcolsep}{2pt}
\begin{tabular}{cc}
\begin{tikzpicture}
\begin{axis}[calibfig, xlabel={independent type $D_L$}, ylabel={dependent type $D_H$},
  xmin=0, xmax=1, ymin=0, ymax=1,
  xtick={0,0.5,1}, xticklabels={$0$,$D^{*}$,$1$},
  ytick={0,0.5,1}, yticklabels={$0$,$D^{*}$,$1$}]
\fill[pattern=north east lines, pattern color=black!30] (axis cs:0,0) -- (axis cs:1,0) -- (axis cs:1,1) -- cycle;
\fill[black!6] (axis cs:0,0) -- (axis cs:0,0.5) -- (axis cs:0.5,0.5) -- cycle;
\fill[black!6] (axis cs:0.5,0.5) -- (axis cs:0.5,1) -- (axis cs:1,1) -- cycle;
\fill[black!20] (axis cs:0,0.5) -- (axis cs:0.5,0.5) -- (axis cs:0.5,1) -- (axis cs:0,1) -- cycle;
\draw (axis cs:0,0) -- (axis cs:1,1);
\draw[dashed] (axis cs:0,0.5) -- (axis cs:0.5,0.5) -- (axis cs:0.5,1);
\addplot[only marks, mark=*, mark size=1.4pt] coordinates {(0.3,0.8)};
\node[font=\tiny, anchor=west] at (axis cs:0.32,0.8) {Table~\ref{tab:params}};
\node[font=\tiny] at (axis cs:0.25,0.93) {separation (trap)};
\node[font=\tiny] at (axis cs:0.64,0.89) {pool on $A$};
\node[font=\tiny, align=center] at (axis cs:0.13,0.37) {pool\\ on $B$};
\node[font=\tiny, align=center, fill=white, inner sep=1.5pt] at (axis cs:0.73,0.25) {excluded:\\ $D_L\ge D_H$};
\end{axis}
\end{tikzpicture}
&
\begin{tikzpicture}
\begin{axis}[calibfig, xlabel={relational stakes $\rho$}, ylabel={value gap $\Delta=w_V-w_N$},
  xmin=0, xmax=2, ymin=0, ymax=2, xtick={0,1,2}, ytick={0,1,2}]
\fill[black!6] (axis cs:0,0) -- (axis cs:1.25,2) -- (axis cs:0,2) -- cycle;
\fill[black!6] (axis cs:0,0) -- (axis cs:2,0) -- (axis cs:2,1.2) -- cycle;
\fill[black!20] (axis cs:0,0) -- (axis cs:2,1.2) -- (axis cs:2,2) -- (axis cs:1.25,2) -- cycle;
\draw[dashed] (axis cs:0,0) -- node[pos=0.8, sloped, above, font=\tiny] {$\Delta=2\rho D_H$} (axis cs:1.25,2);
\draw[dashed] (axis cs:0,0) -- node[pos=0.8, sloped, below, font=\tiny] {$\Delta=2\rho D_L$} (axis cs:2,1.2);
\draw[-{Stealth[length=1.8mm]}, line width=0.6pt] (axis cs:1,0.93) -- (axis cs:1,0.4);
\addplot[only marks, mark=*, mark size=1.4pt] coordinates {(1,1)};
\node[font=\tiny, anchor=east, align=right] at (axis cs:0.97,0.72) {recognition:\\ $\Delta\downarrow$};
\node[font=\tiny] at (axis cs:0.33,1.72) {pool on $B$};
\node[font=\tiny] at (axis cs:1.55,1.72) {separation (trap)};
\node[font=\tiny] at (axis cs:1.55,0.3) {pool on $A$};
\end{axis}
\end{tikzpicture}
\\[-3pt]
{\footnotesize (a) Regimes in type space} & {\footnotesize (b) Regimes in stakes space} \\
\end{tabular}
\caption{Equilibrium regimes (Proposition~\ref{prop:regimes}). \textbf{(a)} For a given threshold $D^{*}=\Delta/2\rho$, the types determine the regime: separation, the trap regime, requires $D_L<D^{*}\le D_H$ (dark); both types pool on $B$ if $D_H<D^{*}$ and on $A$ if $D_L>D^{*}$ (light), in which case messages are uninformative and $N$ is assigned by coin flip. The hatched region is excluded because $D_L<D_H$ by definition. The dot marks the types of Table~\ref{tab:params}. \textbf{(b)} The same partition for the types of Table~\ref{tab:params} ($D_L=0.3$, $D_H=0.8$), in the space of relational stakes $\rho$ and the value gap $\Delta$: separation holds between the rays $\Delta=2\rho D_L$ and $\Delta=2\rho D_H$. The calibrated point $(\rho,\Delta)=(1,1)$ lies in the separation region; recognizing invisible work shrinks $\Delta$ and moves the economy into pooling on accommodation (Proposition~\ref{prop:recognition}).}
\label{fig:regimes}
\end{figure}

\begin{corollary}[Comparative statics of the trap]\label{cor:cs}
On the separation region:
\begin{enumerate}[label=(\alph*),itemsep=1pt]
\item $\partial D^{*}/\partial\rho<0$: higher relational stakes (dependence punishing the appearance of difficulty more) lower the accommodation threshold, so that more types accommodate.
\item $\partial D^{*}/\partial\Delta>0$: when invisible work is more career-damaging relative to visible work, workers accommodate less to avoid the heightened risk; those who still accommodate are the most dependent.
\item The supervisor's gain from sorting is increasing in the friction gap $r_L-r_H$ and vanishes as $r_L-r_H\to0$ (Remark~\ref{rem:condition}). The allocation rule of Lemma~\ref{lem:super}, and hence $q$, is the same for every $r_L>r_H$.
\end{enumerate}
\end{corollary}

\begin{proof}
Parts (a)--(b) follow by differentiating \eqref{eq:threshold}. Part (c) follows from Lemma~\ref{lem:super}: the expected friction saving from placing $N$ on the more dependent worker is $\phi(r_L-r_H)$ times the probability the workers' posteriors differ, which is zero when $r_L=r_H$. The rule in Lemma~\ref{lem:super} depends only on the ranking of posteriors, not on the size of $r_L-r_H$.
\end{proof}

\begin{remark}\label{rem:general}
With a continuum of types $D\sim F$ on $[\underline D,\overline D]$, accommodation remains a cutoff rule $D\ge D^{*}$ with $D^{*}$ solving $\rho D^{*}=q(D^{*})\,\Delta$, where the assignment sensitivity $q(\cdot)$ is now an equilibrium object determined by $F$ and the cutoff. Existence of an interior cutoff follows from continuity and the monotonicity of $\Phi$; the binary model is the special case $q=\tfrac12$.
\end{remark}

\section{Dynamics: The Accommodation Trap}\label{sec:dynamics}

The stage game now repeats, $t=0,1,2,\dots$, and task receipt shapes future dependence. Visible work builds human capital and recognition and so strengthens the outside option; invisible work crowds out career-building time and so weakens it. In terms of dependence,
\begin{equation}\label{eq:lom}
D_{i,t+1}\;=\;D_{i,t}\;+\;\gamma\,\ind\{i\text{ does }N\text{ at }t\}\;-\;\alpha\,\ind\{i\text{ does }V\text{ at }t\},\qquad \alpha,\gamma>0 .
\end{equation}
Doing $V$ lowers dependence by $\alpha$; doing $N$ raises it by $\gamma$. Proposition~\ref{prop:trap} states the result under the stage-by-stage benchmark in which workers play the stage equilibrium each period; Section~\ref{sub:fl} then shows the trap survives \emph{forward-looking} workers under a transparent condition.

\begin{proposition}[Self-reinforcing trap]\label{prop:trap}
Suppose at $t=0$ the workers differ and lie in the separation regime, with $D_{1,0}<D^{*}\le D_{2,0}$ (worker $2$ more dependent). Then along the equilibrium path:
\begin{enumerate}[label=(\roman*),itemsep=1pt]
\item worker $2$ accommodates and is assigned $N$, and worker $1$ is bounded and receives $V$, in every period;
\item the dependence gap grows without bound,
\[
D_{2,t}-D_{1,t}\;=\;\big(D_{2,0}-D_{1,0}\big)+(\alpha+\gamma)\,t \;\uparrow\;\infty;
\]
\item the configuration is absorbing: the inequalities $D_{1,t}<D^{*}\le D_{2,t}$ are preserved for all $t$, so the allocation never reverts.
\end{enumerate}
A transitory difference in dependence at $t=0$ thus becomes permanent, and the worker who began more dependent is locked into invisible work while the other accumulates visible work and outside options.
\end{proposition}

\begin{proof}
By induction. At $t$, suppose $D_{1,t}<D^{*}\le D_{2,t}$. By Proposition~\ref{prop:dep}, worker $2$ (type above threshold) accommodates and worker $1$ does not; by Proposition~\ref{prop:alloc}, the differing-message profile $(B,A)$ assigns $N$ to worker $2$ and $V$ to worker $1$. By \eqref{eq:lom}, $D_{2,t+1}=D_{2,t}+\gamma\ge D_{2,t}\ge D^{*}$ and $D_{1,t+1}=D_{1,t}-\alpha\le D_{1,t}<D^{*}$, so the threshold inequalities are preserved, establishing (i) and (iii). Differencing \eqref{eq:lom} across the two workers gives $D_{2,t+1}-D_{1,t+1}=(D_{2,t}-D_{1,t})+(\alpha+\gamma)$, and iterating from $t=0$ yields (ii).
\end{proof}

\begin{remark}[Bounding the state]\label{rem:bound}
Dependence is realistically bounded; let $D_{i,t}\in[D_{\min},D_{\max}]$ with \eqref{eq:lom} truncated at the endpoints. The divergence in Proposition~\ref{prop:trap} then becomes \emph{absorption}: the trapped worker's type rises to $D_{\max}$ and the favored worker's falls to $D_{\min}$, after which the allocation is permanent. All qualitative conclusions are unchanged, and continuation values stay finite, a property used below.
\end{remark}

\begin{proposition}[Lock-in from equal starts]\label{prop:lockin}
Consider the bounded state of Remark~\ref{rem:bound} with $D_{\min}<D^{*}\le D_{\max}$ and the stage-by-stage play of Proposition~\ref{prop:trap}: each period a worker accommodates if and only if $D_{i,t}\ge D^{*}$, and the supervisor allocates as in Lemma~\ref{lem:super}, breaking ties with a fair coin. From every initial state $(D_{1,0},D_{2,0})$, including $D_{1,0}=D_{2,0}$, the workers reach the configuration of Proposition~\ref{prop:trap}, in which one of them does $N$ and the other does $V$ in every later period, in finite time with probability one. With
\[
k\;\equiv\;\max\Big\{\Big\lceil \frac{D^{*}-D_{\min}}{\gamma}\Big\rceil,\ \Big\lfloor \frac{D_{\max}-D^{*}}{\alpha}\Big\rfloor+1\Big\},
\]
the probability that this configuration has not been reached by period $mk$ is at most $(1-2^{1-k})^{m}$.
\end{proposition}

\begin{proof}
Call a state \emph{separated} if $\min_i D_{i,t}<D^{*}\le\max_i D_{i,t}$. A separated state stays separated and fixes the allocation: the worker at or above $D^{*}$ accommodates and does $N$, so their dependence becomes $\min\{D_{\max},D+\gamma\}\ge D^{*}$, while the worker below $D^{*}$ is bounded and does $V$, so theirs becomes $\max\{D_{\min},D-\alpha\}<D^{*}$. This is Proposition~\ref{prop:trap}(iii) under the bounded law of motion.

A state that is not separated has both workers below $D^{*}$ or both at or above it. They then send the same message, and the coin decides who does $N$. For a given worker $X$, consider the event that in each of the next $k$ periods either the state is already separated or the coin gives $N$ to $X$; it has probability at least $2^{-k}$, and the two events for $X=1$ and $X=2$ are disjoint because the first flip favors only one worker, so their union has probability at least $2^{1-k}$. On this union the state becomes separated within $k$ periods. If both workers start below $D^{*}$, the dependence of $X$ rises by $\gamma$ per period, capped at $D_{\max}\ge D^{*}$, and so reaches $D^{*}$ within $\lceil (D^{*}-D_{\min})/\gamma\rceil\le k$ periods, while the other worker, who does $V$, stays below $D^{*}$. If both start at or above $D^{*}$, the other worker's dependence falls by $\alpha$ per period, floored at $D_{\min}<D^{*}$, and so drops below $D^{*}$ within $\lfloor (D_{\max}-D^{*})/\alpha\rfloor+1\le k$ periods, while $X$ stays at or above $D^{*}$. Hence from every state that is not separated, the probability of separation within the next $k$ periods is at least $2^{1-k}$. By the Markov property, the probability of remaining unseparated through $m$ consecutive blocks of $k$ periods is at most $(1-2^{1-k})^{m}$, which tends to zero as $m\to\infty$.
\end{proof}

\noindent The trap therefore requires no initial difference in dependence: chance alone separates identical workers, and once separated they never trade places.

\subsection{Forward-looking workers}\label{sub:fl}

Proposition~\ref{prop:trap} assumed stage-by-stage play. The natural objection is that a forward-looking worker who foresees the trap might decline to accommodate \emph{today} in order to escape it. The trap is robust to this objection under a transparent condition, and the condition is most easily met for exactly the workers the model is about.

The key is that the state is \emph{sticky}: one period of bounded communication does not restore a worker's outside option. In the trapped configuration, deviating to $B$ turns the profile into $(B,B)$, so the deviator receives the visible task only with probability $\tfrac12$, lowers their dependence by at most $\alpha$, and, if $D\ge D^{*}+\alpha$, remains above the threshold. They are therefore still trapped next period: escape would require a sustained, relationally costly campaign of boundedness, not a one-shot deviation. This makes the static incentive condition close to sufficient.

\begin{proposition}[Forward-looking trap]\label{prop:fl}
Consider the bounded-state dynamic game (Remark~\ref{rem:bound}) with worker discount factor $\delta\in(0,1)$ and the separating trap of Proposition~\ref{prop:trap}, in which the favored worker satisfies $D_{1}+\gamma<D^{*}$, so that one unexpected assignment of $N$ does not lift them over the threshold. Let the trapped worker's dependence satisfy $D\ge D^{*}+\alpha$. If
\begin{equation}\label{eq:fl}
\rho D \;\ge\; \frac{\Delta}{2} \;+\; \underbrace{\frac{\delta}{1-\delta}\,\overline{W}'\,\frac{\alpha+\gamma}{2}}_{\textstyle \Lambda},
\qquad \overline{W}'\equiv W'(O_{\min}),
\end{equation}
then a one-shot deviation to $B$, followed by a return to the trap, lowers the trapped worker's expected discounted payoff, and a one-shot deviation to $A$ lowers the favored worker's. Because the trapped worker's dependence only rises along the trap and the favored worker's only falls, these conditions, once met, hold in every later period, so the trap is robust to one-shot deviations along its entire path. Writing $M\equiv\rho D-\tfrac{\Delta}{2}$ for the static margin and $K\equiv\overline{W}'\,\tfrac{\alpha+\gamma}{2}$, condition \eqref{eq:fl} is equivalent to $\delta\le\bar\delta(D)\equiv M/(M+K)$, a cutoff that rises toward one as the static margin grows (high relational stakes). It therefore holds in the limits $\delta\to0$ (survival pressure) or $\alpha,\gamma\to0$ (gradual human-capital accumulation), and it can fail for patient workers facing fast human-capital returns. When it fails, the trap is no longer guaranteed by this argument; since $\Lambda$ rises with $\delta$ and with $\alpha+\gamma$, the guarantee is weakest for patient workers facing fast human-capital returns.
\end{proposition}

\begin{proof}
Consider a single deviation to $B$ at a trapped state $D\ge D^{*}+\alpha$, after which the worker returns to the trap. The immediate payoff change is $\tfrac{\Delta}{2}-\rho D$: the deviator forgoes the relational benefit $\rho D$ and gains expected task payoff $\tfrac{\Delta}{2}$ from the $(B,B)$ coin flip. After the deviation they revert to accommodation; since $D-\alpha\ge D^{*}$, they remain trapped, and since $D_{1}+\gamma<D^{*}$ the favored worker stays below the threshold even if the coin flip gives them $N$, so the continuation play is unchanged except for a parallel shift of their dependence path downward by an expected $(\alpha+\gamma)/2$ (the difference between the deviation's expected increment $(\gamma-\alpha)/2$ and the on-path increment $\gamma$). Because $W'\le\overline{W}'$ on $[O_{\min},\,\cdot\,]$, the discounted continuation gain is at most $\Lambda$ in \eqref{eq:fl}. The total gain from deviating is therefore at most $\tfrac{\Delta}{2}-\rho D+\Lambda\le0$ under \eqref{eq:fl}. For the favored worker, a one-shot deviation to $A$ turns the profile into $(A,A)$: its immediate effect is $\rho D_1-\tfrac{\Delta}{2}<0$, since $D_1<D^{*}$, and it weakly raises their expected future dependence, which lowers security welfare, while their later play is unchanged because $D_1+\gamma<D^{*}$. Supervisor sequential rationality is Lemma~\ref{lem:super}. Finally, $M\ge\tfrac{\delta}{1-\delta}K$ rearranges to $\delta\le M/(M+K)$, since $M>0$ for $D\ge D^{*}+\alpha$.
\end{proof}

\begin{remark}[What the supervisor observes]\label{rem:history}
The supervisor is assumed to allocate on current messages, as in the stage game, which is natural when outside options also move for reasons she does not see. If she instead tracks history, a worker who has accommodated and done $N$ is already marked as dependent, so a bounded message no longer changes her assignment. The allocation of Proposition~\ref{prop:trap} then persists whatever the later messages, and a trapped worker's deviation only forfeits $\rho D$. The trap is therefore at least as robust as condition \eqref{eq:fl} indicates.
\end{remark}

\begin{remark}[Dependence lowers the effective discount factor]\label{rem:delta}
Condition \eqref{eq:fl} ties the trap to patience, and the link runs through dependence itself. A worker whose job secures survival, status, funding, or insurance discounts the future heavily relative to keeping the current relationship intact; their effective $\delta$ is low, so \eqref{eq:fl} holds easily. The same dependence that makes accommodation rational in the static model makes the dynamic trap self-enforcing. A related link between impatience and separation appears in implicit contracting, where a lower discount factor can move play from pooling to separation \citep{gurtler2008}. A secure worker (high $\delta$, low $\rho$) refuses to be trapped, consistent with Proposition~\ref{prop:dep}. Types within $\alpha$ of the threshold, who \emph{can} escape with a single deviation, are the margin at which the trap breaks.
\end{remark}

\section{Welfare and Correctives}\label{sec:welfare}

\subsection{Is the trap inefficient?}

Each period exactly one $V$ and one $N$ are completed, so there is no static output loss; the supervisor's friction is also minimized by construction. The welfare question is therefore distributional and dynamic. Let a worker's welfare be an increasing, strictly concave function $W$ of their outside option $O_i$ (equivalently, decreasing convex in $D_i$), capturing diminishing returns: an extra unit of outside option is worth more to a worker who has little. A planner who also values worker welfare weighs the supervisor's friction against $\sum_i W(O_i)$.

\begin{proposition}[Inefficiency of the trap]\label{prop:welfare}
Suppose $W$ is strictly concave. Consider reallocating the visible task in some period from the less dependent to the more dependent worker (equivalently, rotating $N$ away from the dependent worker). This raises utilitarian worker welfare $\sum_i W(O_i)$ at a cost to the supervisor of at most $\phi(r_L-r_H)$. Hence there exist parameters ($W$ sufficiently concave or $\phi$ sufficiently small) under which the equilibrium allocation is socially inefficient: a planner would assign the visible task to the more dependent worker, or rotate, and the unregulated supervisor does the opposite.
\end{proposition}

\begin{proof}
By \eqref{eq:lom}, the visible task raises the recipient's outside option by $\alpha$ and the invisible task lowers it by $\gamma$. Moving $V$ from the low-dependence worker (high $O$, hence low $W'(O)$ by concavity) to the high-dependence worker (low $O$, high $W'(O)$), and $N$ the other way, therefore raises the dependent worker's outside option by $\alpha+\gamma$ and lowers the other's by $\alpha+\gamma$, changing utilitarian welfare by approximately $(\alpha+\gamma)\big(W'(O_{\text{dep}})-W'(O_{\text{indep}})\big)>0$ under strict concavity. The only cost to the supervisor is that $N$ now falls on the less dependent worker, raising expected friction by $\phi(r_L-r_H)$. Whenever the welfare gain exceeds this friction cost (guaranteed for $W$ concave enough or $\phi$ small enough), the reallocation raises the planner's objective. The supervisor, who ignores $\sum_i W$, does not undertake it. Thus the equilibrium is inefficient on a nonempty parameter set.
\end{proof}

\noindent The wedge is an externality: the supervisor internalizes friction but not the worker's concave welfare from human-capital accumulation. This is what makes the ``trap'' a genuine inefficiency rather than mere sorting, and it is distributionally regressive: visible work flows to the worker who already has the strongest options.

\subsection{Correctives}

\begin{proposition}[Recognition]\label{prop:recognition}
Attaching recognition or promotability to invisible work mitigates the trap through two channels. (i) Raising $w_N$ toward $w_V$ shrinks $\Delta$, which lowers the threshold $D^{*}=\Delta/2\rho$ and lowers the per-period welfare loss $\Delta$ from misallocation; for $\Delta$ small enough the economy leaves the separation regime for pooling on accommodation (Proposition~\ref{prop:regimes}(iii)), eliminating sorting. (ii) Reducing the crowd-out $\gamma$ in \eqref{eq:lom} flattens the dynamic divergence in Proposition~\ref{prop:trap}, whose per-period increment is $\alpha+\gamma$; the gap stops widening once $\gamma\le-\alpha$. If recognition makes invisible work itself career-building ($\gamma<0$), the trapped worker's dependence falls every period, so they eventually cross below $D^{*}$ and the configuration ceases to be absorbing.
\end{proposition}

\begin{proof}
Channel (i): $\partial D^{*}/\partial\Delta>0$ from \eqref{eq:threshold}, and the per-period payoff gap between tasks is $\Delta$ by definition; as $\Delta\downarrow$, $D^{*}\downarrow$, so eventually $D^{*}<\Dl$ and Proposition~\ref{prop:regimes}(iii) applies. Channel (ii): in the proof of Proposition~\ref{prop:trap}, the gap increment is $(\alpha+\gamma)$; reducing $\gamma$ shrinks it, and it is nonpositive iff $\gamma\le-\alpha$. If $\gamma<0$, then $D_{2,t+1}=D_{2,t}+\gamma<D_{2,t}$, so after finitely many periods $D_{2,t}<D^{*}$ (given $D^{*}>D_{\min}$), and the inequalities $D_{1,t}<D^{*}\le D_{2,t}$ that make the configuration absorbing fail.
\end{proof}

\begin{proposition}[Rotation and reduced discretion]\label{prop:rotation}
Suppose the invisible task is assigned by an exogenous rule (fixed rotation or randomization) rather than by supervisor discretion. Then the assignment is independent of the communication signal, so $q=0$, the net gain from accommodating becomes $\Phi(D)=\rho D>0$ for all types, and accommodation becomes allocation-irrelevant: every type accommodates to collect the relational benefit, but the message no longer affects who does $N$. Invisible work is shared by rule, so the dynamic divergence of Proposition~\ref{prop:trap} does not occur and the dependence gap does not amplify systematically. Committing ex ante to any allocation rule that does not condition on communication achieves the same.
\end{proposition}

\begin{proof}
If $a(\cdot)$ does not depend on $(m_1,m_2)$, then $\Pr(N\mid A)=\Pr(N\mid B)$, so $q=0$ and \eqref{eq:netgain} gives $\Phi(D)=\rho D>0$: every type strictly prefers $A$, but the choice does not affect the allocation. Under a rule that equalizes $N$ across workers over time, each worker does $N$ a fixed fraction of periods; substituting equal task receipt into \eqref{eq:lom} removes the systematic divergence, since neither worker is persistently assigned $N$.
\end{proof}

\begin{remark}[Scope and what the model does not claim]\label{rem:scope}
The model is deliberately minimal and makes no claim of novelty about the existence of invisible work or of non-promotable-task inequality, both of which are established. It does not assume taste-based discrimination, gender, or productivity differences; introducing any of these is orthogonal to the mechanism. The relational benefit $g(D)=\rho D$ is a reduced form (Remark~\ref{rem:microfound}). The contribution is the mechanism (dependence shapes communication, communication shifts the supervisor's perceived friction cost, allocation feeds back onto dependence) and the conditions under which it produces a self-reinforcing trap.
\end{remark}

\begin{remark}[Micro-founding the relational benefit]\label{rem:microfound}
The term $\rho D$ can be derived from an explicit continuation stage in which, after the period, the supervisor decides whether to renew the relationship or supply a favorable reference, and a worker perceived as ``difficult'' faces a higher probability of an unfavorable outcome whose cost to the worker scales with their dependence $D$ (a weaker outside option makes non-renewal more painful). Accommodation lowers that probability. Folding the continuation value into the stage payoff yields a relational benefit increasing in $D$, of which $\rho D$ is the linear case. The main text keeps the reduced form to isolate the signaling mechanism.
\end{remark}

\section{A Numerical Illustration}\label{sec:calibration}

The analysis closes with a parameterized example. The numbers are \emph{illustrative, not estimated}: the paper uses no data, and the example serves only to confirm that the model's conditions are mutually satisfiable and to display the magnitudes it implies. Normalize $w_V=1$, $w_N=0$ (so $\Delta=1$) and $\rho=1$; Table~\ref{tab:params} lists the rest.

\begin{table}[h]
\centering
\begin{tabular}{llc}
\hline
Symbol & Meaning & Value\\
\hline
$\Delta=w_V-w_N$ & value gap, visible vs.\ invisible task & $1$\\
$\rho$ & relational-stakes coefficient & $1$\\
$D_L,\,D_H$ & dependence types & $0.3,\ 0.8$\\
$\pi$ & prior $\Pr(D=D_H)$ & $0.5$\\
$r_L,\,r_H$ & resistance probabilities & $0.4,\ 0.1$\\
$\phi$ & supervisor friction cost & $0.1$\\
$\alpha,\,\gamma$ & visible/invisible effect on dependence & $0.1,\ 0.1$\\
$[D_{\min},D_{\max}]$ & state bounds & $[0.1,\ 0.9]$\\
$O(D)=1-D,\ W(O)=\sqrt{O}$ & outside option; worker welfare & n/a\\
$\delta$ & worker discount factor & $\{0.5,\ 0.9\}$\\
\hline
\end{tabular}
\caption{Illustrative parameter values.}\label{tab:params}
\end{table}

\paragraph{Static equilibrium.} The threshold is $D^{*}=\Delta/2\rho=0.5$, so \SEP\ holds: $D_L=0.3<0.5\le0.8=D_H$. Accommodation propensities are $\Phi(D_H)=\rho D_H-\tfrac{\Delta}{2}=0.3>0$ and $\Phi(D_L)=-0.2<0$, so the dependent type accommodates and the independent type does not. Invisible-work probabilities are $\Pr(N\mid A)=1-\pi/2=0.75$ against $\Pr(N\mid B)=(1-\pi)/2=0.25$.

\paragraph{Inefficiency.} Reallocating the visible task to the dependent worker, and the invisible task to the other, raises one-period welfare by $(\alpha+\gamma)\,[W'(0.2)-W'(0.7)]=0.2\,(1.118-0.598)\approx0.104$, at friction cost $\phi(r_L-r_H)=0.1\times0.3=0.03$. Since $0.104>0.03$, the equilibrium allocation is inefficient here. The cutoff is $\phi^{*}=0.104/0.3\approx0.347$: the allocation is inefficient when $\phi<\phi^{*}$ and efficient otherwise, so the model does not mechanically predict inefficiency (Proposition~\ref{prop:welfare}).

\paragraph{Forward-looking trap.} The immediate no-deviation margin is $\rho D_H-\tfrac{\Delta}{2}=0.3$, and the continuation bound is $\Lambda=\tfrac{\delta}{1-\delta}\,W'(O_{\min})\,\tfrac{\alpha+\gamma}{2}=\tfrac{\delta}{1-\delta}\times1.581\times0.1$. For $\delta=0.5$, $\Lambda\approx0.158<0.3$: condition \eqref{eq:fl} holds, so the trap survives a forward-looking worker's one-shot deviation. For $\delta=0.9$, $\Lambda\approx1.423>0.3$, so condition \eqref{eq:fl} no longer guarantees the trap. The cutoff is $\delta^{*}=\bar\delta(D_H)=0.3/(0.3+0.158)\approx0.655$. Dependent workers, who weight current survival heavily (low effective $\delta$), sit squarely in the region where \eqref{eq:fl} guarantees the trap (Remark~\ref{rem:delta}). The condition is conservative: solving the calibrated dynamic program directly shows that a one-shot deviation from the trap is unprofitable at every trapped state with $D\ge D^{*}+\alpha$ for every $\delta<1$, because with the bounded state the gain from a lower dependence path is temporary and is largely offset by the smaller relational benefit. The favored worker satisfies $D_1+\gamma=0.4<D^{*}$, as Proposition~\ref{prop:fl} requires.

\paragraph{Trajectory.} From an initial gap $D_{2,0}-D_{1,0}=0.5$, the unbounded benchmark of Proposition~\ref{prop:trap} widens the gap by $\alpha+\gamma=0.2$ each period ($0.5,\,0.7,\,0.9,\,1.1,\dots$). With the bounds, the trapped worker caps at $D_{\max}=0.9$ after one period and the favored worker floors at $D_{\min}=0.1$ after two, so the gap goes $0.5,\,0.7,\,0.8$ and then stays at $0.8$; the allocation is permanent.

\paragraph{Lock-in from equal starts.} Proposition~\ref{prop:lockin} applies with $k=5$. Solving the Markov chain on the grid of Table~\ref{tab:params} exactly, two workers who both start at dependence $0.3$ are locked in after $4$ periods in expectation, two who both start at $0.5$ after exactly one period, and two who both start at $0.8$ after $13$ periods in expectation; from every starting point, lock-in occurs with probability one.

\paragraph{The four panels.} Figure~\ref{fig:calib} collects the calibration. Panel~(a) is the trap itself: from the initial types the trapped worker is driven up to $D_{\max}$ and the favored worker down to $D_{\min}$ (solid), while the unbounded benchmark of Proposition~\ref{prop:trap} (dashed) would diverge without limit; the allocation never crosses back over the threshold $D^{*}$. Panel~(b) is the static separation: the accommodation incentive $\Phi(D)=\rho D-\tfrac{\Delta}{2}$ crosses zero at $D^{*}=0.5$, so $D_L$ stays bounded and $D_H$ accommodates. Panel~(c) plots the continuation bound $\Lambda(\delta)$ against the static margin $\rho D_H-\tfrac{\Delta}{2}=0.3$: condition \eqref{eq:fl} guarantees the trap for $\delta\le\delta^{*}\approx0.655$ (shaded), where dependent, survival-focused workers sit, but not at $\delta=0.9$ ($\Lambda\approx1.42$, off the chart); as noted above, the exact computation shows that the trap survives one-shot deviations there too. Panel~(d) plots the per-period welfare gain from reallocation against the supervisor's friction cost $\phi(r_L-r_H)=0.3\phi$: the equilibrium is inefficient for $\phi\le\phi^{*}\approx0.347$ (shaded), and the calibrated $\phi=0.1$ lies inside that region.

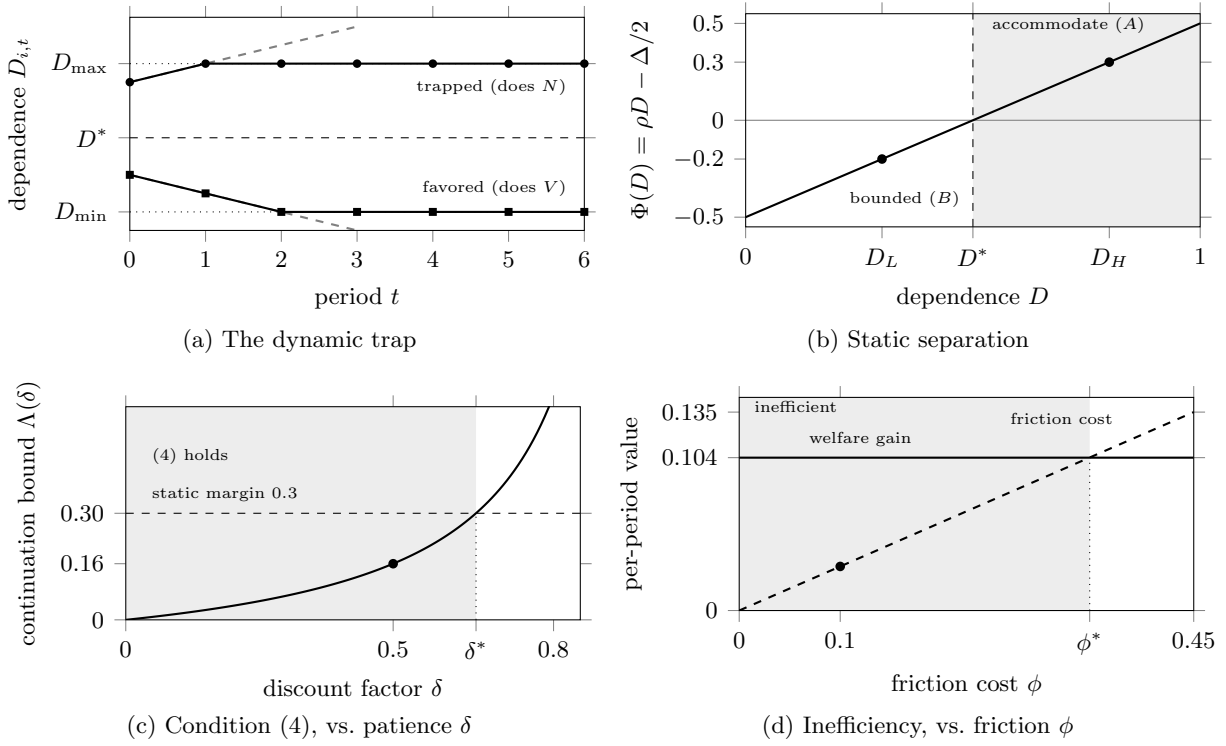
\begin{figure}[t]
\centering
\setlength{\tabcolsep}{2pt}
\begin{tabular}{cc}
\begin{tikzpicture}
\begin{axis}[calibfig, xlabel={period $t$}, ylabel={dependence $D_{i,t}$},
  xmin=0, xmax=6, ymin=0, ymax=1.15, xtick={0,1,2,3,4,5,6},
  ytick={0.1,0.5,0.9}, yticklabels={$D_{\min}$,$D^{*}$,$D_{\max}$}]
\draw[dotted] (axis cs:0,0.9) -- (axis cs:6,0.9);
\draw[dotted] (axis cs:0,0.1) -- (axis cs:6,0.1);
\draw[dashed] (axis cs:0,0.5) -- (axis cs:6,0.5);
\addplot[gray, dashed] coordinates {(0,0.8)(1,0.9)(2,1.0)(3,1.1)};
\addplot[gray, dashed] coordinates {(0,0.3)(1,0.2)(2,0.1)(3,0.0)};
\addplot[mark=*, mark size=1.2pt] coordinates {(0,0.8)(1,0.9)(2,0.9)(3,0.9)(4,0.9)(5,0.9)(6,0.9)};
\addplot[mark=square*, mark size=1.1pt] coordinates {(0,0.3)(1,0.2)(2,0.1)(3,0.1)(4,0.1)(5,0.1)(6,0.1)};
\node[font=\tiny, anchor=north east] at (axis cs:5.9,0.87) {trapped (does $N$)};
\node[font=\tiny, anchor=south east] at (axis cs:5.9,0.13) {favored (does $V$)};
\end{axis}
\end{tikzpicture}
&
\begin{tikzpicture}
\begin{axis}[calibfig, xlabel={dependence $D$}, ylabel={$\Phi(D)=\rho D-\Delta/2$},
  xmin=0, xmax=1, ymin=-0.55, ymax=0.55,
  xtick={0,0.3,0.5,0.8,1}, xticklabels={$0$,$D_L$,$D^{*}$,$D_H$,$1$},
  ytick={-0.5,-0.2,0,0.3,0.5}]
\addplot[draw=none, fill=black!7] coordinates {(0.5,-0.55)(1,-0.55)(1,0.55)(0.5,0.55)} \closedcycle;
\draw[gray] (axis cs:0,0) -- (axis cs:1,0);
\draw[dashed] (axis cs:0.5,-0.55) -- (axis cs:0.5,0.55);
\addplot[domain=0:1, samples=2] {x-0.5};
\addplot[only marks, mark=*, mark size=1.4pt] coordinates {(0.3,-0.2)(0.8,0.3)};
\node[font=\tiny, anchor=north] at (axis cs:0.35,-0.31) {bounded ($B$)};
\node[font=\tiny, anchor=south west] at (axis cs:0.52,0.40) {accommodate ($A$)};
\end{axis}
\end{tikzpicture}
\\[-3pt]
{\footnotesize (a) The dynamic trap} & {\footnotesize (b) Static separation} \\[5pt]
\begin{tikzpicture}
\begin{axis}[calibfig, xlabel={discount factor $\delta$}, ylabel={continuation bound $\Lambda(\delta)$},
  xmin=0, xmax=0.85, ymin=0, ymax=0.6,
  xtick={0,0.5,0.655,0.8}, xticklabels={$0$,$0.5$,$\delta^{*}$,$0.8$},
  ytick={0,0.158,0.3}, yticklabels={$0$,$0.16$,$0.30$}]
\addplot[draw=none, fill=black!7] coordinates {(0,0)(0.655,0)(0.655,0.6)(0,0.6)} \closedcycle;
\draw[dashed] (axis cs:0,0.3) -- (axis cs:0.85,0.3);
\draw[dotted] (axis cs:0.655,0) -- (axis cs:0.655,0.3);
\addplot[domain=0:0.84, samples=120] {0.158113883*x/(1-x)};
\addplot[only marks, mark=*, mark size=1.4pt] coordinates {(0.5,0.158)};
\node[font=\tiny, anchor=west] at (axis cs:0.03,0.46) {\eqref{eq:fl} holds};
\node[font=\tiny, anchor=south west] at (axis cs:0.03,0.305) {static margin $0.3$};
\end{axis}
\end{tikzpicture}
&
\begin{tikzpicture}
\begin{axis}[calibfig, xlabel={friction cost $\phi$}, ylabel={per-period value},
  xmin=0, xmax=0.45, ymin=0, ymax=0.145,
  xtick={0,0.1,0.347,0.45}, xticklabels={$0$,$0.1$,$\phi^{*}$,$0.45$},
  ytick={0,0.104,0.135}, yticklabels={$0$,$0.104$,$0.135$}]
\addplot[draw=none, fill=black!7] coordinates {(0,0)(0.347,0)(0.347,0.145)(0,0.145)} \closedcycle;
\addplot[domain=0:0.45, samples=2] {0.104};
\addplot[domain=0:0.45, samples=2, dashed] {0.3*x};
\draw[dotted] (axis cs:0.347,0) -- (axis cs:0.347,0.104);
\addplot[only marks, mark=*, mark size=1.4pt] coordinates {(0.1,0.03)};
\node[font=\tiny, anchor=south] at (axis cs:0.12,0.105) {welfare gain};
\node[font=\tiny, anchor=south east] at (axis cs:0.38,0.12) {friction cost};
\node[font=\tiny, anchor=south west] at (axis cs:0.005,0.128) {inefficient};
\end{axis}
\end{tikzpicture}
\\[-3pt]
{\footnotesize (c) Condition \eqref{eq:fl}, vs.\ patience $\delta$} & {\footnotesize (d) Inefficiency, vs.\ friction $\phi$} \\
\end{tabular}
\caption{Numerical illustration, parameters of Table~\ref{tab:params}. \textbf{(a)} The dynamic trap: solid lines are bounded dependence paths (Remark~\ref{rem:bound}), the dashed lines the unbounded benchmark of Proposition~\ref{prop:trap}; the dashed horizontal line is the threshold $D^{*}$ and the dotted lines the bounds. \textbf{(b)} The accommodation incentive $\Phi(D)$ is increasing and crosses zero at $D^{*}=0.5$; the dependent type $D_H$ sends $A$, the independent type $D_L$ sends $B$. \textbf{(c)} Condition \eqref{eq:fl} holds where $\Lambda(\delta)\le\rho D_H-\tfrac{\Delta}{2}=0.3$, i.e.\ $\delta\le\delta^{*}\approx0.655$ (shaded); the dependent type ($\delta=0.5$) is inside, the secure type ($\delta=0.9$, $\Lambda\approx1.42$) outside. The condition is sufficient, not necessary: in this calibration a one-shot deviation from the trap is unprofitable for every $\delta<1$. \textbf{(d)} The equilibrium is inefficient where the welfare gain ($\approx0.104$) exceeds the friction cost $\phi(r_L-r_H)=0.3\phi$, i.e.\ $\phi\le\phi^{*}\approx0.347$ (shaded).}
\label{fig:calib}
\end{figure}

\section{Discussion}\label{sec:discussion}

\paragraph{Inequality without malice.} The sharpest reading of the model is that unequal allocation of invisible work can arise with a fully rational, non-prejudiced supervisor and equally productive workers. The supervisor engages in statistical inference on dependence, not favoritism; the worker accommodates rationally to protect a relationship they depend on. Neither behaves badly, yet the outcome is persistent inequality. This both distinguishes the mechanism from taste-based accounts and makes it harder to detect and to remedy: there is no smoking gun, only incentives.

\paragraph{When is the trap strongest?} Corollary~\ref{cor:cs} and Proposition~\ref{prop:trap} imply the trap is most severe when relational stakes are high ($\rho$ large), invisible work strongly crowds out career capital ($\gamma$ large), the supervisor gains more from sorting ($r_L-r_H$ large), and allocation is discretionary. In organizational terms: high supervisor discretion, weak formal rotation, ambiguous job descriptions, informal relational contracts, and workers whose status, funding, insurance, or solvency is tied to the job. These are testable scope conditions.

\paragraph{Relation to accommodation as signaling.} The model formalizes \emph{accommodation signaling}: communication that lowers perceived refusal risk and relational friction. Unlike impression management aimed at supervisor liking, the payoff-relevant consequence here is task allocation, and unlike a fixed disposition, the signal is a deliberate choice whose informativeness is endogenous. The Spencian irony is that the signal is self-harming: it sorts the sender into the worse allocation.

\section{Designs for Future Empirical Testing}\label{sec:empirics}

The paper is theoretical; this section records designs that would discipline it.

\paragraph{Supervisor vignette experiment (tests Proposition~\ref{prop:alloc}).} Participants act as supervisors and allocate one visible and one tedious task between two otherwise-equal workers whose only difference is communication style (bounded vs.\ accommodative). Manipulate, between subjects, the workers' stated outside option, the presence of a rotation rule, and whether invisible work is recognized. The prediction is that the accommodator receives the tedious task more often, that the effect strengthens under high discretion and weakens under rotation/recognition, and that it survives equalizing prior performance.

\paragraph{Survey (tests Proposition~\ref{prop:dep}).} Measure financial, visa, and insurance dependence and a family safety net alongside self-reported communication style and the frequency of invisible-work assignment and public recognition. The prediction is that weaker outside options correlate with more accommodative language and with more invisible-work receipt.

\paragraph{Repeated lab game (tests Proposition~\ref{prop:trap}).} Implement the dynamic game with endowment-based outside options that evolve with task receipt. The prediction is that low-endowment workers develop accommodative strategies and accumulate invisible work, and that the dependence gap widens across rounds.

\paragraph{Communication measurement.} Using consented or simulated workplace messages, construct accommodation markers (length, apologies, gratitude, hedging, availability language, boundary clarity) and relate them to subsequent assignments. Any use of real correspondence requires anonymization, consent, and ethics approval.

\section{Conclusion}\label{sec:conclusion}

This paper has argued that the communication strategy protecting a dependent worker's short-run employment relationship can expose them to long-run inequality in task allocation. Weak outside options make accommodation rational; accommodation lowers a supervisor's perceived cost of assigning invisible work; invisible work erodes future outside options and sustains the dependence that began the cycle. The result is an accommodation trap that emerges without discrimination or error, that locks in even workers who start out identical, and that is inefficient whenever diminishing returns to outside options outweigh the supervisor's friction cost. The signal that sustains it harms the worker who sends it, and because the supervisor's beliefs are correct, the remedy lies in rules and incentives rather than in correcting beliefs: recognition, rotation, and reduced discretion break the trap by severing the link between how a worker communicates and what work they are given. The mechanism is general; the model is a first, deliberately minimal formalization, and the designs above indicate how it could be tested.

\section*{Statements and Declarations}

\paragraph{Competing interests.} The author has no competing interests to declare that are relevant to the content of this article.

\paragraph{Funding.} The author did not receive support from any organization for the submitted work.

\paragraph{Data availability.} Data sharing is not applicable to this article, as no datasets were generated or analysed during the current study.

\appendix
\theoremstyle{plain}
\newtheorem{propA}{Proposition}[section]
\newtheorem{lemA}{Lemma}[section]
\theoremstyle{remark}
\newtheorem{remA}{Remark}[section]

\section{Mathematical Appendix}\label{app:math}

This appendix supplies the derivations compressed in the main text: the assignment probabilities and the robustness of $q>0$ (\S\ref{app:q}); the complete three-regime characterization, including the pooling-on-accommodation case and the off-path beliefs that support each regime (\S\ref{app:regimes}); the full forward-looking condition, deriving the continuation bound $\Lambda$ from primitives and stating a sharper sufficient condition (\S\ref{app:fl}); a micro-foundation of the relational benefit $g(D)=\rho D$ from an explicit renewal subgame (\S\ref{app:micro}); and an existence proof for the continuum-type cutoff (\S\ref{app:cont}). Notation is as in the main text; $\Pr(N\mid m)$ denotes the equilibrium probability that a worker sending message $m$ receives the invisible task.

\subsection{Assignment probabilities and the robustness of \texorpdfstring{$q>0$}{q>0}}\label{app:q}

Fix worker $i$ and let the opponent send $A$ with probability $p\in[0,1]$ (in the separating equilibrium $p=\pi$). Let $\tau\in[0,1]$ be the probability the supervisor assigns $N$ to worker $i$ when the two messages tie and posteriors are equal; the main text uses the fair coin $\tau=\tfrac12$.

\begin{lemA}[Assignment probabilities, general tie-break]\label{lemA:q}
Under Lemma~\ref{lem:super},
\[
\Pr(N\mid A)=p\,\tau+(1-p),\qquad
\Pr(N\mid B)=(1-p)\,\tau,
\]
so that
\[
q\;\equiv\;\Pr(N\mid A)-\Pr(N\mid B)\;=\;p\,\tau+(1-p)(1-\tau).
\]
In particular, under the fair coin $\tau=\tfrac12$ one has $q=\tfrac12$ \emph{for every} $p$; and for any $\tau\in(0,1)$ one has $q>0$. The knife-edge $q=0$ requires $(p,\tau)\in\{(1,0),(0,1)\}$, i.e.\ a degenerate opponent combined with a tie-break that perfectly offsets it.
\end{lemA}

\begin{proof}
Condition on the opponent's message. If $i$ sends $A$: with probability $p$ the profile is $(A,A)$ (equal posteriors, so $i$ gets $N$ with probability $\tau$); with probability $1-p$ the profile is $(A,B)$, and since $A\Rightarrow$ higher posterior dependence, $i$ gets $N$ with probability $1$. Hence $\Pr(N\mid A)=p\tau+(1-p)$. If $i$ sends $B$: with probability $p$ the profile is $(B,A)$, so $N$ goes to the opponent and $i$ gets it with probability $0$; with probability $1-p$ the profile is $(B,B)$, probability $\tau$. Hence $\Pr(N\mid B)=(1-p)\tau$. Subtracting,
\[
q=p\tau+(1-p)-(1-p)\tau=p\tau+(1-p)(1-\tau).
\]
At $\tau=\tfrac12$, $q=\tfrac{p}{2}+\tfrac{1-p}{2}=\tfrac12$. For $\tau\in(0,1)$ both terms are nonnegative and at least one is strictly positive (since $p$ and $1-p$ cannot both vanish), so $q>0$.
\end{proof}

\noindent Lemma~\ref{lemA:q} shows the value $q=\tfrac12$ used in the body is not an artifact of the prior $\pi$: it holds for any opponent accommodation probability under a fair tie-break. What the body's results require is only $q>0$, which holds for all interior tie-breaks. With $n>2$ workers competing for a single invisible task the same argument gives $\Pr(N\mid A)=\sum_{k=0}^{n-1}\binom{n-1}{k}p^{k}(1-p)^{n-1-k}\,\tfrac{1}{k+1}$ and $\Pr(N\mid B)=(1-p)^{n-1}\tfrac1n$, again with $q>0$; the binary case is $n=2$.

\subsection{The three regimes in full}\label{app:regimes}

Write a worker's message payoffs, given own type $D$ and induced assignment probabilities, as
\[
U(A;D)=\rho D+w_V-\Delta\,\Pr(N\mid A),\qquad
U(B;D)=w_V-\Delta\,\Pr(N\mid B),
\]
so the accommodation incentive is
\begin{equation}\label{eqA:incentive}
U(A;D)-U(B;D)=\rho D-\Delta\,q .
\end{equation}
Off-path beliefs are refined by monotonicity: an unexpected $A$ is attributed to the most-dependent type and an unexpected $B$ to the least-dependent type. This is the natural belief here and is consistent with an intuitive-criterion argument, because by \eqref{eqA:incentive} the net benefit of sending $A$ is strictly increasing in $D$, so among types contemplating an out-of-equilibrium $A$ it is the high type whose deviation is (weakly) least unprofitable; symmetrically for $B$. Refinements matter because pooling equilibria are often supported by implausible off-path beliefs. \citet{roperogarcia2025} studies signaling games without single crossing in which several pooling equilibria survive divinity, and shows that neologism-proofness selects the plausible one; in the present model single crossing holds, so the simpler monotone refinement suffices.

\begin{propA}[Three regimes, complete]\label{propA:regimes}
Let $D^{*}=\Delta/2\rho$ as in \eqref{eq:threshold}, and adopt the monotone off-path beliefs above. Then:
\begin{enumerate}[label=\textnormal{(\roman*)},itemsep=1pt]
\item \emph{Separation.} If $\Dl<D^{*}\le\Dh$, the separating profile $\{\Dh\mapsto A,\ \Dl\mapsto B\}$ with the allocation of Lemma~\ref{lem:super} is a PBE.
\item \emph{Pooling on boundedness.} If $D^{*}>\Dh$, the profile in which both types send $B$ is a PBE.
\item \emph{Pooling on accommodation.} If $D^{*}<\Dl$, the profile in which both types send $A$ is a PBE.
\end{enumerate}
\end{propA}

\begin{proof}
\emph{(i)} On the separating path each message occurs with positive probability, so beliefs are Bayesian: $A\Rightarrow\Dh,\ B\Rightarrow\Dl$, giving $q=\tfrac12$ by Lemma~\ref{lemA:q}. By \eqref{eqA:incentive}, type $\Dh$ weakly prefers $A$ iff $\rho\Dh\ge\Delta/2$ and type $\Dl$ strictly prefers $B$ iff $\rho\Dl<\Delta/2$; both hold under $\Dl<D^{*}\le\Dh$. Supervisor optimality is Lemma~\ref{lem:super}. This reproduces Proposition~\ref{prop:dep}.

\emph{(ii)} Suppose both types send $B$. On-path the belief at $B$ is the prior, and the opponent sends $B$ with probability one. A deviation to $A$ is off-path; by the monotone belief it is read as $\Dh$, so against the $B$-opponent the deviator faces profile $(A,B)$ and receives $N$ with probability one. The deviator's payoff is therefore $\rho D+w_V-\Delta$, against the equilibrium payoff $w_V-\tfrac{\Delta}{2}$ from the $(B,B)$ tie. The deviation is unprofitable iff
\[
\rho D+w_V-\Delta\;\le\;w_V-\tfrac{\Delta}{2}
\quad\Longleftrightarrow\quad
\rho D\le\tfrac{\Delta}{2}
\quad\Longleftrightarrow\quad
D\le D^{*}.
\]
This must hold for the most-tempted (highest) type, so the binding requirement is $\Dh\le D^{*}$, i.e.\ $D^{*}\ge\Dh$, which is implied by $D^{*}>\Dh$. Hence both types send $B$ and the deviation is blocked.

\emph{(iii)} Suppose both types send $A$. On-path the belief at $A$ is the prior, and the opponent sends $A$ with probability one. A deviation to $B$ is off-path; by the monotone belief it is read as $\Dl$, so against the $A$-opponent the deviator faces profile $(B,A)$ and the supervisor assigns $N$ to the (more-dependent-looking) $A$-opponent: the deviator receives $N$ with probability zero and obtains $w_V$. The equilibrium payoff is $\rho D+w_V-\Delta\cdot\tfrac12$ from the $(A,A)$ tie. The deviation is unprofitable iff
\[
w_V\;\le\;\rho D+w_V-\tfrac{\Delta}{2}
\quad\Longleftrightarrow\quad
\rho D\ge\tfrac{\Delta}{2}
\quad\Longleftrightarrow\quad
D\ge D^{*}.
\]
This must hold for the least-incentivized (lowest) type, so the binding requirement is $\Dl\ge D^{*}$, i.e.\ $D^{*}\le\Dl$, which is implied by $D^{*}<\Dl$. Hence both types send $A$ and the deviation is blocked. This completes the case left as ``routine'' in the body.
\end{proof}

\noindent The three conditions partition the line: $D^{*}\le\Dl$ (pool on $A$), $\Dl<D^{*}\le\Dh$ (separate), $D^{*}>\Dh$ (pool on $B$). The boundary cases $D^{*}=\Dl$ and $D^{*}=\Dh$ admit both the adjacent pooling and separating profiles, the indifferent type randomizing.

\subsection{The forward-looking condition, derived}\label{app:fl}

This subsection derives \eqref{eq:fl} from primitives. Normalize the outside option $O(D)=\bar O-D$, so $|O'(D)|=1$ and worker flow welfare from security is $W(O(D))$ with $W$ increasing and strictly concave; hence $\frac{d}{dD}W(O(D))=-W'(O)$ and $W'(O)$ is largest at the smallest $O$, i.e.\ $W'(O)\le W'(O_{\min})\equiv\overline W'$. A worker maximizes
\begin{equation}\label{eqA:obj}
\E\sum_{t\ge0}\delta^{t}\,u_t,\qquad
u_t=\underbrace{W(O(D_t))}_{\text{security}}+\underbrace{\rho D_t\,\ind\{m_t=A\}}_{\text{relational}}+\underbrace{w_V-\Delta\,\ind\{N_t\}}_{\text{task}} ,
\end{equation}
with $D_t$ evolving by the bounded law of motion of Remark~\ref{rem:bound}. The relational term is normalized so that the bounded message yields zero. The stage game pins down only the difference $\rho D$ between the two messages, so this normalization is a modeling choice, and it matters for the continuation accounting below: with it, a lower dependence path lowers future relational benefits, a cost of deviating that is dropped here. If instead an accommodator's relational payoff fell with dependence at rate $c>0$, as it would under the levels of the renewal subgame in Appendix~\ref{app:micro}, the bound $\Lambda$ would use $\overline W'+c$ in place of $\overline W'$.

Consider a trapped worker at state $D\ge D^{*}+\alpha$ on the path where they send $A$ and receive $N$ every period (so $D_t$ drifts up by $\gamma$ until the cap $D_{\max}$). Consider a single deviation to $B$ at $D$, after which they revert to the trapped strategy.

\paragraph{Period-of-deviation change.} Sending $B$ instead of $A$ at $D$:
\begin{itemize}[itemsep=1pt,topsep=2pt]
\item forgoes the relational benefit $\rho D$;
\item changes the task lottery from the sure $N$ (payoff $w_V-\Delta$) to the $(B,B)$ tie, expected task payoff $\tfrac12 w_V+\tfrac12(w_V-\Delta)=w_V-\tfrac{\Delta}{2}$, a gain of $+\tfrac{\Delta}{2}$.
\end{itemize}
The net period-of-deviation change is therefore
\begin{equation}\label{eqA:period}
\tfrac{\Delta}{2}-\rho D .
\end{equation}

\paragraph{Continuation change.} On the trapped path the period increment to $D$ is $+\gamma$ (sure $N$). Under the deviation the increment is stochastic: $(B,B)$ gives $N$ with probability $\tfrac12$ (increment $+\gamma$) and $V$ with probability $\tfrac12$ (increment $-\alpha$), so the expected increment is $\tfrac{\gamma-\alpha}{2}$. The deviation thus lowers the period's expected increment by
\[
\gamma-\tfrac{\gamma-\alpha}{2}=\tfrac{\gamma+\alpha}{2}.
\]
Because $D\ge D^{*}+\alpha$ guarantees $D-\alpha\ge D^{*}$, the worker is still trapped next period and resumes the sure-$N$ path (and because the favored worker satisfies $D_1+\gamma<D^{*}$, their message is unchanged too); since the law of motion is additive, this one-time wedge shifts the \emph{entire} future expected $D$-path \emph{downward} by $\tfrac{\alpha+\gamma}{2}$, in every subsequent period until the path reaches the cap $D_{\max}$, which only shrinks the wedge. A downward shift of $D_t$ by $\varepsilon$ raises security flow by $-\frac{d}{dD}W(O(D))\,\varepsilon=W'(O_t)\,\varepsilon\le\overline W'\varepsilon$. Discounting from $t+1$,
\begin{equation}\label{eqA:cont}
\text{continuation welfare gain}\;=\;\sum_{s\ge1}\delta^{s}\,W'(O_{t+s})\,\tfrac{\alpha+\gamma}{2}
\;\le\;\frac{\delta}{1-\delta}\,\overline W'\,\frac{\alpha+\gamma}{2}\;=\;\Lambda .
\end{equation}
The same downward shift also \emph{lowers} future relational benefits $\rho D_{t+s}$ (a cost of deviating), which is omitted; omitting it only strengthens the no-deviation conclusion, which is why the resulting condition is sufficient rather than tight.

\paragraph{Combining.} Adding \eqref{eqA:period} and \eqref{eqA:cont}, the total gain from the deviation is at most $\tfrac{\Delta}{2}-\rho D+\Lambda$, which is $\le0$ exactly when
\[
\rho D\;\ge\;\tfrac{\Delta}{2}+\Lambda,
\qquad
\Lambda=\frac{\delta}{1-\delta}\,\overline W'\,\frac{\alpha+\gamma}{2},
\]
reproducing \eqref{eq:fl}. This proves Proposition~\ref{prop:fl}.

\begin{propA}[A sharper sufficient condition]\label{propA:tight}
Let $\{O_{t+s}\}_{s\ge1}$ be the on-path (truncated) outside-option sequence from the trapped state $D$, and let $g_{t+s}\ge0$ denote the per-period reduction in relational benefit caused by the downward path-shift. The deviation is unprofitable whenever
\[
\rho D\;\ge\;\tfrac{\Delta}{2}\;+\;\frac{\alpha+\gamma}{2}\sum_{s\ge1}\delta^{s}W'(O_{t+s})\;-\;\sum_{s\ge1}\delta^{s}g_{t+s}.
\]
Replacing $W'(O_{t+s})$ by its upper bound $\overline W'$ and dropping the (nonnegative) relational-loss term $\sum\delta^{s}g_{t+s}$ yields the sufficient condition \eqref{eq:fl}; hence \eqref{eq:fl} is conservative, and the true escape margin is weakly smaller.
\end{propA}

\begin{proof}
Immediate from \eqref{eqA:period} and the continuation accounting: the discounted relational loss is $\sum_{s\ge1}\delta^{s}g_{t+s}$, and the discounted security gain is at most $\frac{\alpha+\gamma}{2}\sum_{s\ge1}\delta^{s}W'(O_{t+s})$, because concavity makes the first-order term an upper bound and the wedge only shrinks once the path reaches $D_{\max}$. The displayed inequality is therefore sufficient for the one-shot deviation to be unprofitable, and it is exact to first order in $\alpha+\gamma$ while the wedge persists. The bounds $W'(O_{t+s})\le\overline W'$ and $g_{t+s}\ge0$ give \eqref{eq:fl}.
\end{proof}

\subsection{Micro-founding the relational benefit}\label{app:micro}

This subsection derives $g(D)=\rho D$ from an explicit continuation, formalizing Remark~\ref{rem:microfound}. Append to each stage a \emph{renewal subgame}: after tasks are completed, the relationship either continues (favorable: renewal, good reference, retained goodwill) or ends unfavorably. Let the probability of the favorable outcome depend on the worker's communication style,
\[
\Pr(\text{favorable}\mid A)=p_A,\qquad \Pr(\text{favorable}\mid B)=p_B,\qquad \Delta p\equiv p_A-p_B>0 ,
\]
so accommodation raises the chance of a favorable continuation. Let the worker's loss from the unfavorable relative to the favorable outcome be
\[
L(D)\;=\;\big(\text{value if relationship continues}\big)-\big(\text{value of falling back on the outside option }O(D)\big),
\]
which is increasing in $D$ because a weaker outside option makes the fallback worse; take the linear case $L(D)=\ell\,D$ with $\ell>0$.

\begin{propA}[Reduced form of the relational benefit]\label{propA:micro}
The expected continuation payoff of message $m$ is $-(1-p_m)L(D)$ up to a type-and-message-independent constant. Hence the relational benefit of $A$ over $B$ is
\[
g(D)=\big[-(1-p_A)L(D)\big]-\big[-(1-p_B)L(D)\big]=(p_A-p_B)\,L(D)=\Delta p\,\ell\,D ,
\]
which is exactly $\rho D$ with $\rho=\Delta p\,\ell>0$. More generally, for any increasing $L(\cdot)$ the relational benefit $g(D)=\Delta p\,L(D)$ is increasing in $D$, so the single-crossing property of $\Phi$ in \eqref{eq:netgain} survives; $\rho D$ is the linear specialization.
\end{propA}

\begin{proof}
Normalize the favorable-outcome continuation value to a constant $\bar V$ (independent of $D$ and $m$ given the relationship continues) and the unfavorable continuation to $\bar V-L(D)$. Then the expected continuation under message $m$ is $p_m\bar V+(1-p_m)(\bar V-L(D))=\bar V-(1-p_m)L(D)$. Differencing across $m=A,B$ removes $\bar V$ and yields $g(D)=(p_A-p_B)L(D)$. Substituting $L(D)=\ell D$ gives $g(D)=\Delta p\,\ell\,D=\rho D$. Monotonicity for general increasing $L$ is immediate since $\Delta p>0$.
\end{proof}

\noindent Thus the body's primitive $g(D)=\rho D$ is the linear case of a renewal subgame in which (a) accommodation raises the favorable-continuation probability by $\Delta p$, and (b) the cost of an unfavorable continuation rises with dependence. Both ingredients are the same single force used throughout: dependence raises the cost of friction with the supervisor.

\subsection{Existence with a continuum of types}\label{app:cont}

Let $D\sim F$, atomless, on $[\underline D,\overline D]$, i.i.d.\ across the two workers, and conjecture a cutoff strategy: send $A$ iff $D\ge x$. The following result completes Remark~\ref{rem:general}.

\begin{propA}[Continuum cutoff]\label{propA:cont}
\emph{(a) Two workers, fair tie-break.} If both workers use cutoff $x$ and the supervisor breaks ties fairly, then $q=\tfrac12$ independently of $x$ (Lemma~\ref{lemA:q} with $p=1-F(x)$, $\tau=\tfrac12$), so the indifferent type solves $\rho x=\Delta\,q=\tfrac{\Delta}{2}$, giving the unique cutoff $x^{*}=\dfrac{\Delta}{2\rho}=D^{*}$. An interior separating equilibrium exists iff $\underline D<D^{*}<\overline D$.

\emph{(b) Endogenous assignment sensitivity.} If the supervisor can condition on a finer signal so that the assignment sensitivity $q$ depends on the population cutoff, $q=q(x)$ continuous in $x$, then a cutoff equilibrium solves the fixed point $\rho x=\Delta\,q(x)$. The best-response map $T(x)=\Delta\,q(x)/\rho$ is a continuous self-map of the compact interval $[\underline D,\overline D]$ whenever $T(x)\in[\underline D,\overline D]$, so by Brouwer's theorem a fixed point $x^{*}=T(x^{*})$ exists; it is the equilibrium cutoff. Strict monotonicity of $\Phi(D)=\rho D-\Delta q$ in $D$ makes the implied strategy a genuine cutoff rule.
\end{propA}

\begin{proof}
\emph{(a)} With both workers above-cutoff sending $A$, two $A$-senders share the truncated posterior $F(\cdot\mid D\ge x)$ and are indistinguishable to the supervisor, so they tie and assign $N$ by fair coin; Lemma~\ref{lemA:q} then gives $q=\tfrac12$ for every $x$. The indifference condition $\rho x-\Delta q=0$ becomes $\rho x=\tfrac{\Delta}{2}$, with unique root $D^{*}$; types above accommodate (since $\Phi$ is increasing) and types below do not. Interiority requires $\underline D<D^{*}<\overline D$.

\emph{(b)} $T$ is continuous because $q(\cdot)$ is, and it maps the compact convex set $[\underline D,\overline D]$ into itself by hypothesis; Brouwer yields $x^{*}=T(x^{*})$, i.e.\ $\rho x^{*}=\Delta q(x^{*})$. At $x^{*}$ a worker is indifferent; by single-crossing of $\Phi$ all higher types strictly prefer $A$ and all lower types strictly prefer $B$, confirming the cutoff structure.
\end{proof}

\noindent Part (a) shows that within the two-worker model the continuum cutoff is exactly the binary threshold $D^{*}$; part (b) records that, once the supervisor reads finer information so that $q$ becomes an equilibrium object, existence is a standard fixed-point argument.

\end{document}